\documentclass[12pt,a4paper]{article}
\input xy
\xyoption{all}
\usepackage{amsthm}
\usepackage{amsfonts,amsmath}
\usepackage{lscape}

\def\tp{\otimes} 

\def\d{\delta}

\def\L{\mathcal{L}}
\def\ve{\varepsilon}
\def\s{\sigma}
\def\t{\tau}

\newtheorem{thm}{Theorem}[section]
\newtheorem{prop}[thm]{Proposition}

\theoremstyle{definition}

\theoremstyle{remark}

\title{Bethe ansatz solution of an integrable, non-Abelian anyon chain with $D(D_3)$ symmetry}
\author{
C.W. Campbell, K.A. Dancer, P.S. Isaac and J. Links \\
Centre for Mathematical Physics, \\
School of Physical Sciences,  \\
The University of Queensland, 4072,\\
Australia.
}

\begin{document}
\maketitle

\begin{abstract}
\noindent The exact solution  for the energy spectrum of a one-dimensional Hamiltonian with local two-site interactions and periodic boundary conditions is determined. The two-site Hamiltonians commute with the symmetry algebra given by the Drinfeld double $D(D_3)$ of the dihedral group $D_3$.  As such the model describes local interactions between non-Abelian anyons, with fusion rules given by the tensor product decompositions of the irreducible representations of $D(D_3)$. The Bethe ansatz equations which characterise the exact solution are found through the use of functional relations satisfied by a set of mutually commuting transfer matrices. 
\end{abstract}

\section{Introduction}

Current interest in systems of non-Abelian anyons and associated topological phases is primarily motivated by applications to topological quantum computing \cite{kitaev2003,nssfd2008}. There are indications that non-Abelian excitations occur in  fractional quantum Hall systems with filling fraction 5/2 \cite{g07,s07}. Fractional quantum Hall systems are necessarily chiral in nature, breaking time-reversal symmetry due to the presence of a strong magnetic field.  An early example of an anyonic model with chiral edge modes was given by Kitaev \cite{k06}. There are also expectations that topological phases may   
occur in time-reversal invariant systems \cite{m09,s09}, and models have been developed in this context \cite{lw05,f08,gtkltw09}.
 
To help facilitate a deeper understanding of non-Abelian anyonic systems, several groups have undertaken studies of models of interacting non-Abelian anyons on a one-dimensional lattice \cite{ftltkwf2007,by07,ttwl08,frbm08,a1,a2,fltr09}. Here we continue research in the direction of interacting systems by developing the Bethe ansatz solution for a chiral, integrable chain which was introduced in \cite{DIL2006}. The Hamiltonian of this model was derived using the algebraic structure given by the Drinfeld double \cite{d1986} $D(D_3)$ of the dihedral group $D_3$. Consistent braiding and fusion relations for non-Abelian anyonic models are found in the category of irreducible 
representations of $D(D_3)$, since $D(D_3)$ is an example of a quasi-triangular Hopf algebra 
\cite{kitaev2003,d1986}. Associated with a (non-trivial) quasi-triangular Hopf algebra are braiding properties for anyonic degrees of freedom, characterised by solutions of the Yang-Baxter equation {\it without} spectral parameter which are realised through representations of the universal $R$-matrix of the algebra. The anyonic fusion rules are given by  decompositions of tensor product representations of the Hopf algebra, which are governed by the coproduct structure. These fusion rules provide a means to construct interacting systems by assigning energies to the various possible multiplet structures \cite{ttwl08}. In \cite{DIL2006} an integrable Hamiltonian,  comprised of two-site interactions which commute with the action of $D(D_3)$, was derived.  The Hamiltonian was constructed by means of the Quantum Inverse Scattering Method (QISM) \cite{STF1979} through a solution of the Yang-Baxter equation {\it with} spectral parameter. However the algebraic Bethe ansatz approach to derive the exact solution, which often accompanies the QISM, is problematic for this model due to the apparent lack of a pseudo-vacuum state.

Here we determine the Bethe ansatz solution for this model by the method of functional relations satisfied by a set of mutually commuting transfer matrices, following techniques developed in \cite{br89,bbp90,nep03}. While our results have close connection to some others in the literature, we also find differences. It can be shown \cite{bp09,fdil09} that the $D(D_3)$-invariant solution of the Yang-Baxter equation found in \cite{DIL2006} can be obtained as a limiting case of the checkerboard three-state self-dual Potts model, or equivalently the checkerboard three-state Fateev-Zamolodchikov model \cite{fz82}. However the one-dimensional Hamiltonian derived from this particular limiting case does not belong to the class of chiral Potts chain models derived in \cite{bpa88}, this latter class arising from the choice of uniform rapidities. The non-chiral limit of the chiral Potts chain models gives the Fateev-Zamolodchikov spin models,  the Bethe ansatz solution for which was derived in \cite{a92} using functional relations given in \cite{bbp90}. The energy eigenvalues of these Hamiltonians are functions of a set of roots of a set of Bethe ansatz equations. What we will show below is that the energy eigenvalues of the non-Abelian anyonic chain we study are generically functions of two sets of roots of a single set of Bethe ansatz equations. This situation is reminiscent of the Bethe ansatz results for particular random tiling models \cite{k94,dgn97}.

\section{Preliminaries}

The use of Drinfeld doubles of finite group algebras for describing anyons on a lattice was introduced in \cite{kitaev2003}. Here we recall some relevant facts pertaining to our investigations. 

\vfil\eject

The dihedral group $D_3$ has two generators $\sigma, \tau$ satisfying:

$$ \sigma^3 = e,\; \tau^2 = e,\; \tau \sigma = \sigma^{2} \tau $$

\noindent where $e$ denotes the identity.  The Drinfeld double \cite{d1986} of $D_3$, denoted $D(D_3)$, has basis $$\{ gh^*| g,h \in D_3\},$$ where $g$ are the group elements and $g^*$ are dual elements. This gives an algebra of dimension 36. Multiplication of dual elements is given by 

\begin{eqnarray}
g^* h^* = \delta(g,h) g^*
\label{rels1}
\end{eqnarray}

\noindent where $\d$ is the Kronecker delta function.  The products $h^*g$ are computed using

\begin{eqnarray}
h^*g = g(g^{-1} h g)^*.
\label{rels2}
\end{eqnarray}

The algebra $D(D_3)$ becomes a Hopf algebra by imposing the following coproduct, antipode and counit respectively:

\begin{align*}
&\Delta(gh^*) = \sum_{k \in G} g (k^{-1}h)^* \otimes g k^* = \sum_{k \in G} gk^* \otimes g(h k^{-1})^*, \\
&S(gh^*) = (h^{-1})^*g^{-1} = g^{-1} (g h^{-1}g^{-1})^*, \\
&\ve(gh^*)= \d(h,e),\qquad \qquad\qquad \qquad \qquad   \forall g,h \in D_3.
\end{align*}
The universal $R$-matrix is given by 

$$ \mathcal{R} = \sum_{g \in D_3} g \otimes g^* .$$

\noindent This can easily be shown to satisfy the defining relations for a quasi-triangular Hopf algebra:
\begin{eqnarray}
\mathcal{R}{\Delta}(a)&=&{\Delta}^T(a)\mathcal{R}, \quad \forall\,a\in D(G), \label{qt1}\\
({\Delta}\otimes {\rm id}) \mathcal{R}&=&\mathcal{R}_{13}\mathcal{R}_{23},   \label{qt2}\\
({\rm id}\otimes {\Delta})\mathcal{R}&=& \mathcal{R}_{13}\mathcal{R}_{12}, \label{qt3}
\end{eqnarray}  
where ${\Delta}^T$ is the opposite coproduct
$$ {\Delta}^T(gh^*) = \sum_{k \in G} gk^* \otimes g (k^{-1}h)^* = \sum_{k \in G}
g(kh^{-1})^* \otimes   gk^*. $$

 The Hopf algebra $D(D_3)$ has two 1-dimensional irreducible representations (irreps), four 2-dimensional irreps and two 3-dimensional irreps.  Letting $E^i_j$ denote the matrix with 1 in the $(i,j)$ position and zeros elsewhere, the irreps are given by \cite{DIL2006}:

\begin{itemize}
\item[{}] \underline{1-dimensional irreps} \\

$$\pi_{(1,\pm)} = 1, \quad \pi_{(1,\pm)}(\t) = \pm 1, \quad \pi_{(1,\pm)}(g^*) = \delta(g,e).$$
The representation $\pi_{(1,+)}$ is known as the trivial representation and is associated with the $D(D_3)$ action on the vacuum. It can equivalently be expressed as 
\begin{eqnarray*}
\pi_{(1,+)}(a) = \ve(a) \qquad \forall \,a\in D(D_3).
\end{eqnarray*}

\item[{}] \underline{2-dimensional irreps} \\
Set $\omega$ to be the cube root of unity $\omega = e^{2\pi i/3}$. Then 

$$\pi_{(2,e)}(\s) = \begin{pmatrix}
                    \omega & 0  \\ 0 & \omega^{-1}
                 \end{pmatrix}, \quad
\pi_{(2,e)}(\tau) = \begin{pmatrix}
                 0 & 1 \\ 1 & 0 
               \end{pmatrix}, \quad 
\pi_{(2,e)}(g^*) = \delta(g, e) I_2$$

\noindent and 

$$\pi_{(2,i)}(\s) = \begin{pmatrix}
                    \omega^i & 0  \\ 0 & \omega^{-i}
                 \end{pmatrix}, \quad
\pi_{(2,i)}(\tau) = \begin{pmatrix}
                 0 & 1 \\ 1 & 0 
               \end{pmatrix}, \quad 
\pi_{(2,i)}(g^*) =  \delta(g, \s)E^1_{1}+ \delta(g,\s^{-1}) E^2_{2}$$

\noindent for $i = 0,1,2.$

\item[{}] \underline{3-dimensional irreps} \\

$$\pi_{(3,\pm)}(\s) = \begin{pmatrix}
           0 & 1 & 0 \\ 0 & 0 & 1 \\ 1 & 0 & 0 
	 \end{pmatrix}, \quad
\pi_{(3,\pm)}(\t) = \pm \begin{pmatrix}
             1 & 0 & 0 \\ 0 & 0 & 1 \\ 0 & 1 & 0
	   \end{pmatrix}, $$
$$\pi_{(3,\pm)}((\s^i)^*)=0, \quad \pi_{(3,\pm)}((\s^i \t)^*)=E^{i+1}_{i+1},\quad i =0,1,2.$$
\end{itemize}
Throughout we denote the module associated with an irrep $\pi_{(a,b)}$ by $V_{(a,b)}$.

Tensor product representations of $D(D_3)$ can be constructed through the coproduct action $\Delta$.  An important aspect in the definition of a Hopf algebra is that the coproduct 
is coassociative, meaning $(\rm{ id} \otimes \Delta)\Delta =(\Delta \otimes \rm{ id})\Delta$. This allows for a consistently defined 
action of $D(D_3)$ on any $\L$-fold tensor product space. 
All representations of $D(D_3)$ are completely reducible \cite{gould93}. 
It can be shown that the $\L$-fold tensor product representation of $ \pi_{(3,+)} $ decomposes as 
\begin{eqnarray*}
\pi_{(3,+)}^{\otimes \L} =\begin{cases} \frac{1}{2}\left(3^{\L-1} +1 \right)\pi_{(3,+)} 
\oplus \frac{1}{2}\left(3^{\L-1} -1 \right)\pi_{(3,-)} ,  \qquad\qquad\qquad\quad \L\,\,{\rm odd},  \\
3^{\L -2}\left( \pi_{(2,e)} \oplus \pi_{(2,0)} \oplus \pi_{(2,1)} \oplus \pi_{(2,2)}    \right) \\ 
\qquad\qquad \oplus\frac{1}{2}\left(3^{\L-2} +1 \right)\pi_{(1,+)} 
\oplus \frac{1}{2}\left(3^{\L-2} -1 \right)\pi_{(1,-)} ,  \qquad \L\,\,{\rm even}
\end{cases}
\end{eqnarray*}
where the co-efficients above denote the multiplicities. For the model we will examine below, the module $V_{(3,+)}$ will be used for the local Hilbert space for the anyonic degrees of freedom. The global Hilbert space $W$ is given by the $\L$-fold tensor product
\begin{eqnarray}
W=V_{(3,+)}^{\otimes \L} . 
\label{tpm}
\end{eqnarray}
However for anyonic systems, it is necessary to distinguish the {\it physical} space of states $\mathcal{W}\subset W$ which is governed by a superselection rule \cite{k06}. Here we take the following superselection rule to define $\mathcal{W}$: 
\begin{eqnarray}
\pi_{(3,+)}^{\otimes \L}(a){{w}}=\ve(a) {w} \qquad \,\forall \,{w}\in\mathcal{W},\quad a\in D(D_3),
\label{triv}
\end{eqnarray}   
which is equivalent to demanding that the system of $D(D_3)$ anyons with local states represented by $V_{(3,+)}$ are quasi particles produced from the vacuum. Mathematically, 
$\mathcal{W}$ is the $(3^{\L-2}+1)/2$-dimensional space of trivial modules contained in the tensor product module (\ref{tpm}) when $\L$ is even. When $\L$ is odd  $\mathcal{W}$ is zero-dimensional.

Finally we make the following important technical observation. Let $\{v_1,\,v_2,\,v_3\}$ denote the basis for $V_{(3,+)}$ with respect to the action of the $E^i_j$ defined above, i.e. $E^i_j v^k=\delta^k_j v^i$. Then the space $V_{(1,+)}\subset V_{(3,+)} \otimes V_{(3,+)}$ is spanned by the vector 
\begin{eqnarray}
s= \sum_{k=1}^3 v^k \otimes v^k.
\label{s}
\end{eqnarray}  
Note that $s$ is symmetric even though the coproduct action $\Delta$  is not cocommutative, i.e. $\Delta\neq\Delta^T$.

\section{The Hamiltonian and integrability}

We will study an integrable Hamiltonian, acting on $W$, which reads \cite{DIL2006}
\begin{eqnarray}
H=\sum_{k=1}^{\L-1} h_{k(k+1)} + h_{\L 1}
\label{Ham}
\end{eqnarray}
where 
\begin{eqnarray}
 h=\sum_{\gamma \in D_3}i (E^{\gamma(1)}_{\gamma(2)} \tp E^{\gamma(2)}_{\gamma(3)} 
	- E^{\gamma(2)}_{\gamma(3)} \tp E^{\gamma(1)}_{\gamma(2)} )
\label{ham}
	\end{eqnarray}
and the elements $\gamma\in D_3$ are written as permutations of $\{1,2,3\}$. (Recall that $D_3$ is isomorphic to the permutation group $S_3$.) Since 
$D(D_3)$ is not cocommutative, the inclusion of the periodic boundary interaction $h_{\L 1}$ breaks the global $D(D_3)$ invariance of $H$ down to $D_3$. Nonetheless, we will show below that the physical space of states $\mathcal{W}$ is stable under the action of $H$, and that the global $D(D_3)$ invariance is preserved when we restrict the action of $H$ to only physical states.   

We begin by discussing some simple properties of the Hamiltonian, which is first of all seen to be hermitian.  Defining the permutation operator
$P \in \mbox{End }( V_{(3,+)}\tp V_{(3,+)})$ through
$$ P(v \tp w) = w \tp v, \hspace{1cm} v,w \in V_{(3,+)} $$
it can be checked that (\ref{Ham}) commutes with the translation operator 
$$\mathcal{T}=P_{1\L}P_{1\L-1}...P_{12}=\exp(i\mathcal{P})$$
where $\mathcal{P}$ is the momentum operator.
The chiral nature of (\ref{Ham}) can be seen through the space inversion of sites $i\mapsto \L+1-i$, which induces the maps $\mathcal{P}\mapsto -\mathcal{P}$ and $H \mapsto -H$. That the Hamiltonian is not invariant with respect to space inversion is a signature of chirality.  It is also apparent that the Hamiltonian is not time-reversal invariant as $H^*=-H$, where 
throughout we will use $*$ to denote complex conjugation. Given any eigenstate $|\Psi\rangle$, the time-reversed state $|\Psi\rangle^*$ has energy and momentum which take the negative values of those for $|\Psi\rangle$. As a result we can conclude that both the momentum {\it{and}} energy spectra are symmetrically distributed around zero.  

By construction, each of the local Hamiltonians $h_{i(i+1)}$ for $i=1,...,\L-1$ commutes with the action $D(D_3)$ \cite{DIL2006}. So for any $a\in D(D_3)$ and $ {w} \in{\mathcal{W}}$ we have 
\begin{eqnarray*}
\pi^{\otimes \L}_{(3,+)}(a) h_{i(i+1)} {w} &=& h_{i(i+1)} \pi^{\otimes \L}_{(3,+)}(a) {w} \\
&=&  h_{i(i+1)} \ve(a) {w} \\
&=& \ve(a) h_{i(i+1)} {w}
\end{eqnarray*} 
which shows that $\mathcal{W}$ is stable under the action of  $h_{i(i+1)}$ for $i=1,...,\L-1$. It remains to consider the action of $h_{\L 1}$. We can write any ${w}\in \mathcal{W}$ as 
\begin{eqnarray*}
{w} &=&\sum_{k=1}^3 v^k \otimes u^k \\
&=&\sum_{k=1}^3  t^k \otimes v^k  
\end{eqnarray*}
where, as a consequence of (\ref{s}),  $\{u^k\}_{k=1}^3\subset \left(V_{(3,+)}\right)^{\otimes \L-1}$ and $\{t^k\}_{k=1}^3\subset \left(V_{(3,+)}\right)^{\otimes \L-1}$ are both isomorphic to $\{v^k\}_{k=1}^3\subset V_{(3,+)}$ as a basis for the module associated with the representation $\pi_{(3,+)}$. Using this result we have 
\begin{eqnarray*}
\mathcal{T}{w} = \sum_{k=1}^3 v^k \otimes t^k 
\end{eqnarray*}
which again by comparison with (\ref{s}) shows that $\mathcal{T}{w}\in\mathcal{W}$, and necessarily $\mathcal{T}^{-1}{w}\in\mathcal{W}$. 
Now we have 
\begin{eqnarray*}
h_{\L 1} {w}  &=& \mathcal{T}^{-1} \mathcal{T} h_{\L 1} {w} \\
&=& \mathcal{T}^{-1} h_{12}  \mathcal{T} {w}. 
\end{eqnarray*}
Since $\mathcal{W}$ is stable under the actions of $h_{12},\,\mathcal{T}$ and $\mathcal{T}^{-1}$, this establishes that $\mathcal{W}$ is stable under the action of $h_{\L 1}$. Consequently $\mathcal{W}$ is stable under the action of $H$. As the action of $D(D_3)$ is trivial on $\mathcal{W}$, through the superselection rule (\ref{triv}), we have that the action of $H$ restricted to $\mathcal{W}$ is both stable and $D(D_3)$-invariant.

Next we recall the basic results concerning the integrability of (\ref{Ham}).
We consider an invertible operator $R(x) \in \mbox{End }( V_{(3,+)}\tp V_{(3,+)})$ which satisfies the Yang--Baxter equation in $\mbox{End }(V_{(3,+)} \tp V_{(3,+)} \tp V_{(3,+)})$:
\begin{equation}
 R_{12}(x/y) R_{13}(x) R_{23}(y) = R_{23}(y) R_{13}(x) R_{12}(x/y).
\label{ybe}
\end{equation}
The argument $x$ of $R(x)$ is called the spectral parameter. The explicit solution related to (\ref{Ham}) is \cite{DIL2006}
 \begin{equation}  R(x) = 
  \left( \begin {array}{ccccccccc}
  a(x)&0&0&0&0&0&0&0&0  \\
  0&0& b(x) & c(x)&0&0&0&d(x)&0 \\
  0& b(x) &0&0&0&d(x)& c(x) &0&0 \\
  0& c(x) &0&0&0& b(x) &d(x)&0&0 \\
  0&0&0&0&a(x)&0&0&0&0 \\
  0&0&d(x)& b(x) &0&0&0&c(x)&0 \\
  0&0&c(x)&d(x)&0&0&0& b(x) &0 \\
  0&d(x)&0&0&0& c(x) &b(x) &0&0 \\
  0&0&0&0&0&0&0&0&a(x)
\end {array} \right) 
\label{rm}
\end{equation} 
where 
\begin{eqnarray*}
a(x)&=& x^2-x+1,\\
b(x)&=& x^2-1,\\
c(x)&=& x, \\
d(x)&=&1-x.
\end{eqnarray*}
From this solution we construct the transfer matrix 
\begin{eqnarray}
T(x)={\rm tr}_0(R_{0\L}(x)R_{0(\L-1)}(x)...R_{02}(x)R_{01}(x))
\label{tm}
\end{eqnarray}
which as a result of (\ref{ybe}) forms a commuting family for different values of the spectral parameters:
\begin{eqnarray*}
\left[T(x),\,T(y)\right]=0\qquad\forall x,y\in\mathbb{C}. 
\end{eqnarray*}
Note that $T(1)=\mathcal{T}$. The Hamiltonian is defined by 
\begin{eqnarray} 
H=\left.iT^{-1}(x)\frac{d}{dx}\left[x^{-\L}T(x)\right]\right|_{x=1}.
\label{logderiv}
\end{eqnarray}
This definition give rise to the global Hamiltonian (\ref{Ham}) as the sum of local two-site Hamiltonians (\ref{ham}) where
\begin{eqnarray*}
h=\left.i\frac{d}{dx}\frac{PR(x)}{x}\right|_{x=1}.
\end{eqnarray*} 
By construction the global Hamiltonian satisfies 
\begin{eqnarray*}
\left[H,\,T(x)\right]=0.
\end{eqnarray*}
As a result, $T(x)$ can be used as a generating function for conserved operators of $H$.  

In order to determine the Bethe ansatz solution of the Hamiltonian, it is necessary to consider an expanded set of generating functions for the conserved operators. We introduce the operators
\begin{eqnarray*}
\overline{R}(x)&=&(x^2+x+1)\left[R^{-1}(-x)\right]^{t_2} , \\
\overline{\overline{R}}(x)&=&\frac{(x^2+x+1) (x-1)^2}{ (x^2-x+1) }\left[\overline{R}^{-1}(-x)\right]^{t_2}
\end{eqnarray*} 
where $t_2$ denotes partial matrix transposition in the second space. We also note 
$$R(x)=(1-x)^2(x^2+x+1)\left[\overline{\overline{R}}^{-1}(-x)\right]^{t_2}. $$ 
It follows from (\ref{ybe}) that
\begin{eqnarray}
R_{12}(x/y) \overline{R}_{13}(x) \overline{R}_{23}(y)& =& \overline{R}_{23}(y) \overline{R}_{13}(x) R_{12}(x/y), 
\label{ybe1}\\
R_{12}(x/y) \overline{\overline{R}}_{13}(x) \overline{\overline{R}}_{23}(y)& =& \overline{\overline{R}}_{23}(y) \overline{\overline{R}}_{13}(x) R_{12}(x/y), 
\label{ybe2}\\
\overline{R}_{12}(x/y) \overline{\overline{R}}_{13}(x) \overline{R}_{23}(y)& =& \overline{R}_{23}(y) \overline{\overline{R}}_{13}(x) \overline{R}_{12}(x/y), 
\label{ybe3}\\
\overline{R}_{12}(x/y) {{R}}_{13}(x) \overline{\overline{R}}_{23}(y)& =& \overline{\overline{R}}_{23}(y) {{R}}_{13}(x) \overline{R}_{12}(x/y),
\label{ybe4} \\
\overline{R}_{12}(x/y) {\overline{R}}_{13}(x) {{R}}_{23}(y)& =& {{R}}_{23}(-y) {\overline{R}}_{13}(x) \overline{R}_{12}(x/y).
\label{ybe5}
\end{eqnarray}
Explicitly we have 
\begin{eqnarray*}  \overline{R}(x) &=& 
  \left( \begin {array}{ccccccccc}
  a(x)&0&0&0&c(x)&0&0&0&c(x)  \\
  0&0& b(x) & 0&0&0&0&d(x)&0 \\
  0& b(x) &0&0&0&d(x)& 0 &0&0 \\
  0& 0 &0&0&0& b(x) &d(x)&0&0 \\
  c(x)&0&0&0&a(x)&0&0&0&c(x) \\
  0&0&d(x)& b(x) &0&0&0&0&0 \\
  0&0&0&d(x)&0&0&0& b(x) &0 \\
  0&d(x)&0&0&0& 0 &b(x) &0&0 \\
  c(x)&0&0&0&c(x)&0&0&0&a(x)
\end {array} \right), \\
\overline{\overline{R}}(x) &=& 
  \left( \begin {array}{ccccccccc}
  {x}^{2}+1&0&0&0&x&0&0&0&x \\
  0&0&x^2&-x&0&0&0&1&0 \\
  0&x^2&0&0&0&1&-x&0&0 \\
  0&-x&0&0&0&x^2&1&0&0 \\
  x&0&0&0&{x}^{2}+1&0&0&0&x \\
  0&0&1&x^2&0&0&0&-x&0 \\
  0&0&-x&1&0&0&0&x^2&0 \\
  0&1&0&0&0&-x&x^2&0&0 \\
  x&0&0&0&x&0&0&0&{x}^{2}+1
\end {array} \right). 
\end{eqnarray*} 

Next we introduce the $L$-operator $L(x)\in \mbox{End }( V_{(2,1)}\tp V_{(3,+)})$ which was constructed in \cite{fdil09,dl09} and satisfies the equation
\begin{eqnarray}
L_{12}(x/y) L_{13}(x) R_{23}(y) = R_{23}(y) L_{13}(x) L_{12}(x/y)
\label{llr}
\end{eqnarray}
in $\mbox{End }(V_{(2,1)} \tp V_{(3,+)}\tp V_{(3,+)})$. 
It is convenient to express the solution for the $L$-operator as
\begin{eqnarray*}
L(x)= \left(
\begin{array}{cc} L^1_1(x) & L^1_2(x) \cr 
L^2_1(x) & L^2_2(x) 
\end{array}
\right) 
\end{eqnarray*} 
where
\begin{eqnarray*}
L^1_1(x)&=&xX, \qquad
L^1_2(x)=Z, \\
L^2_1(x)&=&Z^{-1}, \qquad
L^2_2(x)=xX^{-1} 
\end{eqnarray*}
and
\begin{eqnarray*}
X=\left(
\begin{array}{ccc}
0 & 0 & 1 \\
1 & 0 & 0 \\
0 & 1 & 0
\end{array}   
\right), 
\qquad
Z=\left( 
\begin{array}{ccc}
1 & 0 & 0 \\
0 & \omega & 0 \\
0 & 0 & \omega^{-1}
\end{array}
\right).
\end{eqnarray*}
Using $XZ=\omega^{-1} ZX$
we find the inverse to be given by 
\begin{eqnarray}
L^{-1}(x)=\frac{1}{1-\omega^{-1} x^2} \left(
\begin{array}{cc} L^2_2(-\omega^{-1} x) & L^1_2(-\omega^{-1} x) \cr 
L^2_1(-\omega x) & L^2_2(-\omega^{-1} x) 
\end{array}
\right) 
\label{linv}
\end{eqnarray} 
and moreover 
\begin{eqnarray}
[L^{-1}(x)]^{t_2}=\frac{1}{1-\omega^{-1}x^2}L(-\omega^{-1}x).
\label{cross}
\end{eqnarray}
From (\ref{cross}) we obtain
\begin{eqnarray}
L_{12}(x/y) L_{13}(\omega^{-1}x) \overline{R}_{23}(y) 
&= &\overline{R}_{23}(y) L_{13}(\omega^{-1}x) L_{12}(x/y),
\label{llr1} \\
L_{12}(x/y) L_{13}(\omega x) \overline{\overline{R}}_{23}(y) 
&= &\overline{\overline{R}}_{23}(y) L_{13}(\omega x) L_{12}(x/y).
\label{llr2}
\end{eqnarray}
The $L$-operator also satisfies \cite{fdil09,dl09}
\begin{eqnarray}
r_{12}(x/y) L_{13}(x) L_{23}(y) = L_{23}(y) L_{13}(x) r_{12}(x/y)
\label{rrl}
\end{eqnarray}
in $\mbox{End }(V_{(2,1)} \tp V_{(2,1)}\tp V_{(3,+)})$
where $r(x)\in \mbox{End }( V_{(2,1)}\tp V_{(2,1)})$ is a specialisation of the symmetric six-vertex solution:
\begin{eqnarray*}
r(x)=\left( 
\begin{array}{cccc}
\omega x- \omega^{-1}x^{-1} & 0 & 0 & 0 \\
0 & x- x^{-1} & \omega - \omega^{-1} & 0 \\
0 & \omega - \omega^{-1} &  x- x^{-1} & 0 \\
0 & 0 & 0 & \omega x- \omega^{-1}x^{-1} 
\end{array}
\right).
\end{eqnarray*}

Now we construct the additional transfer matrices 
\begin{eqnarray}
t(x)&=&{\rm tr}_0(L_{0\L}(x)L_{0(\L-1)}(x)...L_{02}(x)L_{01}(x)), 
\label{tm1}\\
\overline{T}(x)&=&{\rm tr}_0(\overline{R}_{0\L}(x)\overline{R}_{0(\L-1)}(x)
...\overline{R}_{02}(x)\overline{R}_{01}(x)), 
\label{tm2} \\
\overline{\overline{T}}(x)&=&{\rm tr}_0(\overline{\overline{R}}_{0\L}(x)\overline{\overline{R}}_{0(\L-1)}(x)
...\overline{\overline{R}}_{02}(x)\overline{\overline{R}_{01}}(x)).
\label{tm3}
\end{eqnarray}

As a result of equations 
(\ref{ybe},\ref{ybe1},\ref{ybe2},\ref{ybe3},\ref{ybe4},\ref{ybe5},\ref{llr},\ref{llr1},\ref{llr2},\ref{rrl}) it follows that the set of transfer matrices (\ref{tm},\ref{tm1},\ref{tm2},\ref{tm3}) are mutually commuting for all values of the spectral parameters. This means that this set of transfer matrices can be simultaneously diagonalised and their mutual eigenvectors will be independent of the spectral parameters. Note that $T(x),\overline{T}(x),\overline{\overline{T}}(x)$
are all matrices of polynomials with real co-efficients, so their eigenvalues will either be real polynomials or arise as complex-conjugate pair polynomials for $x\in\mathbb{R}$.  
To conclude this section we prove some similarly useful properties for $t(x)$: 
\begin{prop} \label{prop}
 The transfer matrix $t(x)$ satisfies $t(-x)=(-1)^\L t(x)$, and is self-adjoint for $x\in\mathbb{R}$.
\end{prop}
\begin{proof}
Expressing the transfer matrix as 
\begin{eqnarray}
t(x)=\sum_{i_1,...,i_\L} L^{i_1}_{i_2}(x)\otimes L^{i_2}_{i_3}(x)\otimes
....\otimes L^{i_\L}_{i_1}(x)
\label{expand}
\end{eqnarray}
the first part of the proposition follows from the fact that $L^i_i(x)$ is linear in $x$ while $L^i_j(x)$ is independent of $x$ when $i\neq j$. That $t(x)$ is self-adjoint for $x\in\mathbb{R}$ follows from
\begin{eqnarray*}
\left[t^i_j(x)\right]^\dagger=t^{p(i)}_{p(j)}(x)
\end{eqnarray*}
where $p(1)=2,\,p(2)=1$, and as a result hermitian conjugation leaves the right hand side of (\ref{expand}) invariant. 
\end{proof}

\section{Functional relations and the Bethe ansatz solution}

The key observation needed in formulating a set of functional relations for the transfer matrices is that $L(\omega^{-1})$ is singular, as can be deduced from (\ref{linv}). It can be verified by direct calculations that the nullspace of $L(\omega^{-1})$ is three-dimensional. A specialisation of (\ref{llr}) gives 
\begin{eqnarray*}
L_{12}(\omega^{-1}) L_{13}(x) R_{23}(\omega x) = R_{23}(\omega x) L_{13}(x) L_{12}(\omega^{-1})
\end{eqnarray*}
which shows that the three-dimensional left-nullspace of  $L_{12}(\omega^{-1})$ is invariant under the right action of $R_{23}(\omega x) L_{13}(x) $.
For the tensor product of the first two vector spaces labelled 1 and 2, we adopt  a symmetry-adapted basis consisting of the three left-nullspace spanning vectors of $L(\omega^{-1})$ and three vectors not in the left-nullspace. We then find that we can write
\begin{eqnarray}
 R_{23}(\omega x) L_{13}(x)= \left(
\begin{array}{cc}
(\omega^{-1}x+1)\overline{R}(x) & \star \\
0 & (\omega x-1)\overline{\overline{R}}(\omega^{-1}x) 
\end{array}
\right)
\label{fusion}
\end{eqnarray}  
where $\star$ denotes an expression whose precise form is not needed. Through use of 
(\ref{fusion}) we deduce that the associated transfer matrices satisfy the functional relation
\begin{eqnarray}
T(\omega x)t(x) = f(x)\overline{T}(x)+g(x)\overline{\overline{T}}(\omega^{-1}x)  
\label{funct1}
\end{eqnarray}
where we have set $f(x)=(\omega^{-1}x+1)^\L$ and $g(x)=(\omega x-1)^\L$.
Making the restriction $x\in\mathbb{R}$ and taking the complex conjugate of (\ref{funct1}) leads to a second functional relation
\begin{eqnarray}
T(\omega^{-1} x) \left[t(x)\right]^*= f(\omega^{-1}x)\overline{T}(x)+g(\omega x)\overline{\overline{T}}(\omega x).  
\label{funct2}
\end{eqnarray}
Starting with (\ref{llr1},\ref{llr2}) and performing similar calculations leads to another four functional relations
\begin{eqnarray*}
\overline{T}(\omega x) t(\omega^{-1} x)&=& f(\omega^{-1}x)\overline{\overline{T}}(x)+g(\omega^{-1}x){{T}}(\omega^{-1}x), 
\label{funct3} \\
\overline{T}(\omega^{-1} x) \left[t(\omega^{-1}x)\right]^*&=& f(x)\overline{\overline{T}}(x)+g(\omega^{-1} x){{T}}(\omega x),   
\label{funct4} \\
\overline{\overline{T}}(\omega x) t(\omega x)&=& f(\omega x){T}(x)+g(\omega x)\overline{T}(\omega^{-1}x),  
\label{funct5} \\
\overline{\overline{T}}(\omega^{-1} x) \left[t(\omega x)\right]^*&=& f(\omega x){T}(x)+g( x)\overline{T}(\omega x).
\label{funct6}
\end{eqnarray*} 
Let $\lambda(x),\,\Lambda(x),\overline{\Lambda}(x),\overline{\overline{\Lambda}}(x)$ denote the eigenvalues of $t(x),\,T(x),\overline{T}(x),\overline{\overline{T}}(x)$ respectively, when acting on the global Hilbert space $W$ (\ref{tpm}). Keeping in mind that $t(x)$, $T(x)$, $\overline{T}(x)$, $\overline{\overline{T}}(x)$ are simultaneously diagonalisable, we obtain functional relations for the eigenvalues which can be conveniently expressed in matrix form 
\begin{eqnarray}
\left(
\begin{array}{ccc}
 \lambda(\omega^{-1} x)& -f(\omega^{-1} x)& -g(\omega^{-1} x)\\ 
 -g(  x)& \lambda(x) & -f( x)\\
 -f(\omega x)& -g(\omega x) & \lambda(\omega x)
 \end{array}
 \right) 
 \left(
 \begin{array}{c}
 {{\Lambda}}(x)\\ \overline{\Lambda}(\omega^{-1} x) \\ \overline{\overline{\Lambda}}(\omega x)  
  \end{array}
  \right) 
  = \left(
  \begin{array}{c}
  0\\ 0 \\ 0 
  \end{array}
  \right) 
  \label{matrix1}\\
\left(
\begin{array}{ccc}
 \left[\lambda(\omega^{-1} x)\right]^*& -f(x)& -g(\omega^{-1} x)\\ 
 -g(\omega  x)& \left[\lambda(x)\right]^* & -f( \omega^{-1} x)\\
 -f(\omega x)& -g( x) & \left[\lambda(\omega x)\right]^*
 \end{array}
 \right) 
 \left(
 \begin{array}{c}
 {\Lambda}( x)  \\ \overline{\Lambda}(\omega x) \\  \overline{\overline{\Lambda}}( \omega^{-1} x)  
  \end{array}
  \right) 
  = \left(
  \begin{array}{c}
  0\\ 0 \\ 0 
  \end{array}
  \right).
  \label{matrix2}
\end{eqnarray}
In view of Proposition \ref{prop} the {\it sets} of eigenvalues for $t(x)$ and $[t(x)]^*$ are equivalent. To avoid confusion we have used the notations $\lambda(x)$ and $[\lambda(x)]^*$ in (\ref{matrix1},\ref{matrix2}) to make clear the distinction between a {\it particular} eigenvalue of $t(x)$ compared to a {\it particular} eigenvalue of $[t(x)]^*$.

To find the Bethe ansatz solution we make use of the facts that the matrices appearing in (\ref{matrix1},\ref{matrix2}) must have zero determinant and that their nullspaces are  one-dimensional, which can be shown to be generically true.
Making the change of variable $x \mapsto \omega x$ in (\ref{matrix1}), and rearranging the rows, we obtain
\begin{eqnarray*}
\left(
\begin{array}{ccc}
 \lambda(\omega^{-1} x)& -f(\omega^{-1} x)& -g(\omega^{-1} x)\\ 
 -g(  x)& \lambda(x) & -f( x)\\
 -f(\omega x)& -g(\omega x) & \lambda(\omega x)
 \end{array}
 \right) 
 \left(
 \begin{array}{c}
 \overline{\overline{\Lambda}}(\omega^{-1} x)  \\ \Lambda(\omega x) \\ \overline{\Lambda}( x)  
  \end{array}
  \right) 
  = \left(
  \begin{array}{c}
  0\\ 0 \\ 0 
  \end{array}
  \right).
\end{eqnarray*}
Since the nullspace is  one-dimensional, we have for some rational function $M(x)$
\begin{eqnarray*}
\left(
 \begin{array}{c}
 \Lambda(x) \\ \overline{\Lambda}(\omega^{-1} x) \\ \overline{\overline{\Lambda}}(\omega x) 
  \end{array}
  \right) 
  = M(x) 
  \left(
 \begin{array}{c}
 \overline{\overline{\Lambda}}(\omega^{-1} x)  \\ \Lambda(\omega x) \\ \overline{\Lambda}( x)  
  \end{array}
  \right) 
  = M(x)M(\omega x) 
  \left(
 \begin{array}{c}
 {\overline{\Lambda}}(\omega x) \\ \overline{\overline{\Lambda}}( x)  \\ \Lambda(\omega^{-1} x) \\   
  \end{array}
  \right)
  \end{eqnarray*}
which allows us to express the nullspace spanning vector as
\begin{eqnarray*}
\left(
 \begin{array}{c}
 \Lambda(x) \\ M(x) {\Lambda}(\omega x) \\ M(x) M(\omega x){\Lambda}(\omega^{-1} x) 
  \end{array}
  \right) 
  \end{eqnarray*}
  with $M(x)M(\omega x) M(\omega^{-1} x)=1$. We may now write
\begin{eqnarray*}
\lambda(x)&=& f(x) \frac{M(\omega x) \Lambda (\omega^{-1} x)}{\Lambda(\omega x)} 
+g(x) \frac{\Lambda(x)}{M(x)\Lambda(\omega x) } \\
&=& f(x)z_1(x) +g(x)z_2(x)
\end{eqnarray*}
where
\begin{eqnarray*}
z_1(x)&=& \frac{M(\omega x) \Lambda (\omega^{-1} x)}{\Lambda(\omega x)}, \qquad
z_2(x)= \frac{\Lambda(x)}{M(x)\Lambda(\omega x) }. 
\end{eqnarray*}
It follows that
\begin{eqnarray} 
z_1(x)z_2(\omega x)&=& 1, \label{z1} \\
z_1(x)z_1(\omega x)z_1(\omega^{-1}x)&=&1, \label{z2} \\
z_2(x)z_2(\omega x)z_2(\omega^{-1}x)&=&1. \label{z3}
\end{eqnarray}
It can verified that (\ref{z1},\ref{z2},\ref{z3}) are sufficient conditions for the matrix in (\ref{matrix1}) to have zero determinant. Note that $z_1(x)$ and $z_2(x)$ are rational functions where the polynomials in the numerator and denominator have the same order. The most general form for these functions satisfying (\ref{z1},\ref{z2},\ref{z3}) is
\begin{eqnarray*}
z_1(x) &=& \prod_{j=1}^N \frac{\omega x-iy_j}{x-iy_j},\qquad
z_2(x) = \prod_{j=1}^N \frac{\omega^{-1} x-iy_j}{x-iy_j}
\end{eqnarray*}
for some parameters $S=\{y_j\}_{j=1}^N$, from which an explicit expression for $\lambda(x)$ is found to be
\begin{eqnarray}
\lambda(x)=(\omega^{-1}x+1)^\L \prod_{j=1}^N \frac{\omega x-iy_j}{x-iy_j}
+(\omega x-1)^\L \prod_{j=1}^N \frac{\omega^{-1} x-iy_j}{x-iy_j} .   
\label{lambda}
\end{eqnarray}
Since $\lambda(x)$ as given by (\ref{lambda}) must be a polynomial, setting the residues of $\lambda(x)$ equal to zero leads to the Bethe ansatz equations 
\begin{eqnarray}
\left(\frac{i\omega  y_j -1}{i\omega^{-1} y_j+1}\right)^\L
=\omega\prod^N_{k\neq j}\frac{\omega y_j-y_k}{\omega^{-1}y_j-y_k}
\label{bae}
\end{eqnarray} 
when $y_j\neq 0$. For $y_j=0$ there is no Bethe ansatz equation since the vanishing residue condition is automatically satisfied. 

A feature of the Bethe ansatz equations (\ref{bae}) is that they admit spurious solutions which do not give valid eigenvalue expressions for $\lambda(x)$ when substituted into (\ref{lambda}). (Spurious solutions are also found for Fateev-Zamolodchikov spin models \cite{a92}). Consider the case $N=1$ for which the Bethe ansatz equations (\ref{bae}) reduce to
$$\left(\frac{i\omega y_j -1}{i\omega^{-1}y_j+1}\right)^\L 
=\omega.
$$ 
The above admits solutions 
\begin{eqnarray}
y=-i\frac{ \omega^{1/2\L}\exp(j\pi i/\L) +\omega^{-1/2\L}\exp(-j\pi i/\L) }{\omega^{1-1/2\L}\exp(-j\pi i/\L) - \omega^{1/2\L-1}\exp(j\pi i/\L) },  \qquad j=1,..,\L. 
\label{solns}
\end{eqnarray}
In all cases except $\L=1,\,2$, it can be verified that substituting (\ref{solns}) into (\ref{lambda}) generally leads to an expression which is not real, in contradiction of Proposition \ref{prop}. 

Invoking Proposition \ref{prop} further  we calculate from (\ref{lambda})
\begin{eqnarray}
\lambda(x)&=& (-1)^\L \lambda(-x) \nonumber \\
&=& (\omega^{-1} x-1)^\L \prod_{j=1}^N \frac{\omega x+iy_j}{x+iy_j}
+(\omega x+1)^\L \prod_{j=1}^N \frac{\omega^{-1} x+iy_j}{x+iy_j} 
\label{lambdaminusx}
\end{eqnarray}  
and 
\begin{eqnarray}
\lambda(x)&=& \left[\lambda(x)\right]^* \nonumber \\
&=& (\omega x+1)^\L \prod_{j=1}^N \frac{\omega^{-1} x+iy^*_j}{x+iy^*_j}
+(\omega^{-1} x-1)^\L \prod_{j=1}^N \frac{\omega x+iy^*_j}{x+iy^*_j}, \qquad x\in\mathbb{R }.
\label{lambdastar}\end{eqnarray} 
It is seen that (\ref{lambdaminusx}) and (\ref{lambdastar}) will be equivalent if the set of Bethe ansatz roots satisfying (\ref{bae}) is invariant under complex conjugation.  This strongly suggests  that the roots of (\ref{bae}) which give valid expressions for the eigenvalues $\lambda(x)$ through (\ref{lambda}) will be real or arise in complex conjugate pairs. From numerical solutions of (\ref{bae}) we find that this is the case, with specific examples given in Table 1. For $\L=1,2,3$, a complete set of solutions for $\lambda(x)$ is obtained from the Bethe ansatz equations, in agreement with the results of computational  diagonalisation.  Note that in Table 1 and hereafter we use a two-index notation, $y_{aj}$, for the roots of (\ref{bae}). The first index $a$ labels the set to which the root belongs. The second index $j$ enumerates the elements of each set, which allows us to express the sets of roots as  
\begin{eqnarray*}
 S_a=\{y_{aj}\}_{j=1}^{N_a}.
\end{eqnarray*}

\begin{table}
\begin{center}
\begin{tabular}{|c|c|c|}
\hline
&$\lambda(x)$& Bethe roots $S_a$ \\
\hline
\hline 
 $\L=1$& $-x$&$S_1=\{\,\}$  \\

& $2x$& $S_2=\{y_{21}=-1/\sqrt{3}\}$ \\
\hline
\hline
$\L=2$   &$-x^2-1$    &$S_1=\{y_{11}=0\}$ \\
   &$-x^2+2$   &$S_2=\{y_{21}=-\sqrt{3}/{2}\}$ \\
   &$2x^2-1$   & $S_3=\{y_{31}=0,\,y_{32}=-2/\sqrt{3}\}$\\
   &$2x^2+2$   &$S_4=\{y_{41}=y_{42}^*=-\sqrt{1/3}+i\sqrt{2/3}\}$ \\
\hline
\hline
$\L=3$ &$-x^3-2\sqrt{3}\sin(2\pi/9)  x$ &$S_1=\{y_{11}=1.3500,\,y_{12}=-0.8423\}$ \\
 &$-x^3-2\sqrt{3}\sin(8\pi/9)  x$ &$S_2=\{y_{21}=y_{22}^*=0.3889+0.7687i\}$ \\
 & $-x^3+2\sqrt{3}\sin(4\pi/9)  x$&$S_3=\{y_{31}=-0.3576,\,y_{32}=-1.1043\}$ \\
  & $2x^3$&$S_4=\{y_{41}=y_{42}^*=0.2136+0.6230i,\,y_{43}=-1.3048\}$ \\
   &$2x^3+3x$ &$S_5=\{y_{51}=y_{52}^*=0.6704+1.0129i,\,y_{53}=-3.9134\}$ \\
     &$2x^3-3x$ &$S_6=\{y_{61}=0.7779,\,y_{62}=0.5077,\,y_{63}=-1.4619\}$ \\
   \hline
\end{tabular} 
\end{center}
\caption{Expressions for the transfer matrix eigenvalues $\lambda(x)$ and  associated solutions of the Bethe ansatz equations for $\L=1,\,2,\,3$. For all cases  the results are exact except the Bethe roots with $\L=3$ for which   numerical approximations are given. Agreement is found between the Bethe ansatz results and computational  diagonalisation. } 
\end{table}

Next we define the quantities 
\begin{eqnarray}
\Lambda_{ab}(x) &=& c_{ab}\xi_a(\omega^{-1}x)\left[\xi_b(\omega^{-1}x)\right]^* 
\label{L1}\\
\overline{\Lambda}_{ab}(x) &=& c_{ab}\xi_a(\omega x)\left[\xi_b(\omega x)\right]^* 
\label{L2}\\
\overline{\overline{\Lambda}}_{ab}(x) &=& c_{ab}\xi_a( x)\left[\xi_b(x)\right]^* 
\label{L3}  
\end{eqnarray}
where $c_{ab}$ is some constant and   
\begin{eqnarray*}
\xi_a(x)=\prod_{j=1}^{N_a}(x-iy_{aj}).
\end{eqnarray*}
It can be verified by direct substitution and use of (\ref{lambda},\ref{lambdastar}) that the quantities 
(\ref{L1},\ref{L2},\ref{L3}) satisfy equations (\ref{matrix1},\ref{matrix2}). Moreover we have checked for the cases $\L=1,2,3$ that (\ref{L1},\ref{L2},\ref{L3}) reproduce the transfer matrix eigenvalues which are obtained by computational  diagonalisation, with appropriate choices for the constants $c_{ab}$.  For the case $\L=3$ explicit values for these constants are given in Table 2, as well as the multiplicities of the eigenvalues of $\Lambda_{ab}(x)$, accounting for all 27 eigenvalues. We remark that the degeneracies of $\Lambda_{ab}(x),\,\overline{\Lambda}_{ab}(x),\,\overline{\overline{\Lambda}}_{ab}(x)$ are the same for a given $c_{ab}$. The occurrence of these degeneracies can be understood in terms of the symmetry of the problem, which is given by $D_3$ since the imposition of periodic boundary conditions breaks the global $D(D_3)$ symmetry. (Recall that for this example the physical space of states is trivial since $\L$ is odd.) The three-fold tensor product of a 3-dimensional irreducible representation of $D(D_3)$ decomposes into nine copies of 3-dimensional representations. Each of these representations further decomposes into a sum of 1- and 2-dimensional representations with respect to the $D_3$ subalgebra of $D(D_3)$. This leads in total to nine 1-dimensional and nine 2-dimensional irreducible representations of $D_3$, which are in one-to-one correspondence with the results of Table 2.        

\begin{table}
\begin{center}
\begin{tabular}{|c|c|c|c|}
\hline  
$a$ & $b$ & $c_{ab}$ & multiplicity of $\Lambda_{ab}(x)$ \\
   \hline   
\hline
1 & 1 & $-3+2\sqrt{3}\sin(2\pi/9)$ &  2  \\
1 & 2 & $-2\sqrt{3}\sin(8\pi/9)$    & 2   \\
1 & 3 & $-2\sqrt{3}\sin(2\pi/9)$ &  2  \\
2 & 1 & $-2\sqrt{3}\sin(8\pi/9)$    & 2   \\
2 & 2 & $-3+2\sqrt{3}\sin(8\pi/9)$ &  2  \\
2 & 3 & $2\sqrt{3}\sin(4\pi/9)$    & 2 \\
3 & 1 & $-2\sqrt{3}\sin(2\pi/9)$ &  2  \\
3 & 2 & $2\sqrt{3}\sin(4\pi/9)$    & 2   \\
3 & 3 & $3+2\sqrt{3}\sin(4\pi/9)$ &  2  \\
4 & 4 & 3 &  1  \\
4 & 5 & 3 &   1 \\
4 & 6 & -3 &   1 \\
5 & 4 & 3 &   1 \\
5 & 5 & 3 &  1  \\
5 & 6 & -3 & 1   \\
6 & 4 & -3 &  1 \\
6 & 5 & -3 & 1   \\
6 & 6 & 3 &  1  \\     
\hline 
\end{tabular} 
\end{center}
\caption{Choices of the constants $c_{ab}$ in (\ref{L1}) for $\L=3$ which reproduce the eigenvalues $\Lambda_{ab}(x)$ (not shown) found through computational  diagonalisation. The sets $S_a$ which are  used to evaluate $\Lambda_{ab}(x)$ are those in Table 1. The multiplicities obtained by computational  diagonalisation are given in the final column, the sum of which confirms that the full spectrum is reproduced through (\ref{L1}). }
\end{table} 

We have also confirmed that (\ref{L1},\ref{L2},\ref{L3}) hold true for $\L=1,\,2$, and expect that these formulae will give the spectrum of $\Lambda(x)$ for all $\L$ with real values for 
$c_{ab}$. In general,  assuming that the constants $c_{ab}$ are real ensures that the eigenvalues of 
$\Lambda_{ab}(x)$, $\overline{\Lambda}_{ab}(x)$, $\overline{\overline{\Lambda}}_{ab}(x)$ will either be real or occur as complex conjugate pairs, as required. An unexpected feature of these results is that in many cases the transfer matrix eigenvalues are dependent on 
two set of roots of the Bethe ansatz equations (\ref{bae}). However, not all combinations of $a$ and $b$ give expressions in the spectrum of the transfer matrices.

As a final comment, we indicate that the constants $c_{ab}$ are not determined through the Bethe ansatz equations (\ref{bae}). Fortunately, it is not necessary to know the $c_{ab}$ in order to compute the energy spectrum of (\ref{Ham}). From the definition (\ref{logderiv}) we obtain that the energy eigenvalues are given by  
\begin{eqnarray}
E_{ab}&=&\left.i\Lambda_{ab}^{-1}(x)\frac{d}{dx}\left[x^{-\L}\Lambda_{ab}(x)\right]\right|_{x=1} \nonumber \\
&=&i\left(-\L+\sum_{j=1}^{N_a}\frac{\omega^{-1}}{\omega^{-1}-iy_{aj}}
+\sum_{k=1}^{N_b}\frac{\omega}{\omega+iy_{bk}^*}  \right)
\label{nrg}
\end{eqnarray} 
which are independent of $c_{ab}$.

\section{Conclusion} 

We have calculated the expression (\ref{nrg}), which is a function of two sets of roots of the Bethe ansatz equations (\ref{bae}), as the energy eigenvalues of the Hamiltonian (\ref{Ham}) acting on the Hilbert space $W$.
We have illustrated that the equations (\ref{bae}) admit spurious solutions, and that not all combinations of roots in (\ref{nrg}) are needed to generate the spectrum. A problem to address in the future is to determine the structure of roots which lead to legitimate expressions for the spectrum of $H$ on $W$, and furthermore to classify the solutions which give the spectrum of $H$ restricted to $\mathcal{W}$.  It would also be useful to characterise the appropriate root structures which correspond to the ground-state and elementary excitations in this latter case.   
  
Another direction for future work is to extend the techniques developed here to chains with different boundary conditions. Integrable boundary conditions for this model on an open chain have been determined in \cite{dfil09}, and it is also possible to construct integrable closed chains with non-local (i.e. braided) boundary interactions between the 1st and $\L$th sites 
\cite{lf}. 

\section*{Acknowledgements} K.A.D. acknowledges
the support of the Australian Research Council under Discovery Project DP1092513. P.S.I. is supported by an Early
Career Researcher Grant from The University of Queensland.


\begin{thebibliography}{50}

\bibitem{kitaev2003} A. Yu Kitaev, {\it  
Fault-tolerant quantum computation by anyons}, Ann. Phys. {\bf 303}, 2 (2003). 

\bibitem{nssfd2008}
C. Nayak, S.H. Simon, A. Stern, M. Freedman, and S. Das Sarma, {\it Non-Abelian anyons and topological quantum computation},  
Rev. Mod. Phys. {\bf 80}, 1083 (2008).

\bibitem{g07} V.J. Goldman, {\it Fractional quantum Hall effect:
a game of five halves}, Nature Physics {\bf 3}, 517 (2007).

\bibitem{s07} K. Shtengel, {\it Fractional exchange statistics:
a home for anyon?}, Nature Physics {\bf 3}, 763 (2007). 

\bibitem{k06} A. Kitaev, {\it Anyons in an exactly solvable model and beyond}, Ann. Phys. {\bf 321}, 2 (2006).

\bibitem{m09} J. Moore, {\it Topological insulators: the next generation}, Nature Physics {\bf 5}, 378 (2009).

\bibitem{s09} K. Schoutens, {\it Topological phases: wormholes in quantum matter}, Nature Physics {\bf 5}, 784 (2009).

\bibitem{lw05} M.A. Levin and X.-G. Wen, {\it String-net condensation: a physical mechanism for topological phases} Phys. Rev. B {\bf 71}, 045110 (2005).

\bibitem{f08} P. Fendley, {\it Topological order from quantum loops and nets},  Ann. Phys. {\bf 323}, 3113 (2008). 
 
\bibitem{gtkltw09} C. Gils, S. Trebst, A. Kitaev, A.W.W. Ludwig, M. Troyer,
and Z. Wang,
{\it Topology-driven quantum phase transitions in
time-reversal-invariant anyonic quantum liquids}, Nature Physics {\bf 5}, 834 (2009).

\bibitem{ftltkwf2007}
A. Feiguin, S. Trebst, A.W.W. Ludwig, M. Troyer, A. Kitaev, Z. Wang, and M.H. Freedman,
{\it Interacting anyons in topological quantum liquids: the golden chain}, Phys. Rev. Lett. {\bf 98}, 160409 (2007). 

\bibitem{by07} N.E. Bonesteel and K. Yang, {\it Infinite-randomness fixed points for chains of non-Abelian quasiparticles}, Phys. Rev. Lett. {\bf 99}, 140405 (2007). 

\bibitem{ttwl08} 
S. Trebst, M. Troyer, Z. Wang, A.W.W. Ludwig, {\it A short introduction to Fibonacci anyon models},  
Prog. Theor. Phys. Supp. {\bf 176}, 384 (2008). 

\bibitem{frbm08}  L. Fidkowski, G. Refael, N.E. Bonesteel, and J.E. Moore,
{\it c-theorem violation for effective central charge of infinite-randomness fixed points},
Phys. Rev. B {\bf 78}, 224204 (2008). 

\bibitem{a1} S. Trebst, E. Ardonne, A. Feiguin, D.A. Huse, A.W.W. Ludwig and M. Troyer, 
{\it  Collective states of interacting Fibonacci anyons}, Phys. Rev. Lett. {\bf 101},   149901 (2008).  

\bibitem{a2} C. Gils, E. Ardonne, S. Trebst, A.W.W Ludwig, M. Troyer and Z.H. Wang, 
{\it  Collective states of interacting anyons, edge states, and the nucleation of topological liquids}, Phys. Rev. Lett. {\bf 103},  099903 (2009). 

\bibitem{fltr09} L. Fidkowski, H.-H. Lin, P. Titum, and G. Refael, {\it Permutation-symmetric critical phases in disordered non-Abelian anyonic chains}, Phys. Rev. B {\bf 79}, 155120 (2009). 

\bibitem{DIL2006} K.A. Dancer, P.S. Isaac and J. Links, {\it Representations of the quantum doubles of finite group algebras and spectral parameter dependent solutions of the Yang-Baxter equation}, J. Math. Phys. {\bf 47}, 103511 (2006).

\bibitem{d1986} V.G. Drinfeld,  {\it Quantum groups} in Proceedings of the International Congress of
Mathematicians, A.M. Gleason (ed.) pp. 798-820 (Providence, RI: American Mathematical
Society, 1986).  

\bibitem{STF1979} E.K. Sklyanin, L.A. Takhtadzhyan, and L.D. Faddeev, {\it Quantum inverse problem method. 1}, Theor. Math. Phys. {\bf 40}, 688 (1979). 

\bibitem{br89} V.V. Bazhanov and N.Yu. Reshetikhin, {\it Critical RSOS models and conformal field-theory}, Int. J. Mod. Phys. A {\bf 4}, 115 (1989).

\bibitem{bbp90} R.J. Baxter, V.V. Bazhanov, and J.H.H. Perk, {\it Functional relations for transfer matrices of the chiral Potts model}, Int. J. Mod. Phys. B {\bf 4}, 803 (1990).

\bibitem{nep03} R.I. Nepomechie, {\it Functional relations and Bethe ansatz for the XXZ chain}, J. Stat. Phys. {\bf 111}, 1363 (2003).

\bibitem{bp09} V.V. Bazhanov and J.H.H. Perk, {\it private communication} (2009).

\bibitem{fdil09} P.E. Finch, K.A. Dancer, P.S. Isaac, and J. Links, {\it Solutions of the Yang-Baxter equation: descendants
of the six-vertex model from the Drinfeld doubles of dihedral group algebras}, arXiv:1003.0501 (2010).

\bibitem{fz82} V. Fateev and A.B. Zamolodchikov, {\it Self-dual solutions of the star-triangle relations in $\mathbb{Z}_N$ models}, Phys. Lett. A {\bf 92}, 37 (1982).

\bibitem{bpa88} R.J. Baxter, J.H.H. Perk, and H. Au-Yang, {\it New solutions of the star-triangle relations for the chiral Potts model}, Phys. Lett. A {\bf 128}, 138 (1988). 

\bibitem{a92} G. Albertini, {\it Bethe-ansatz type equations for the Fateev-Zamolodchikov spin model}, J. Phys. A: Math. Gen. {\bf 25}, 1799 (1992).

\bibitem{k94} P.A. Kalugin, {\it The square-triangle random-tiling model in the thermodynamic limit},  J. Phys. A: Math. Gen. {\bf 27}, 3599 (1994).

\bibitem{dgn97} J. de Gier and B. Nienhuis, {\it The exact solution of an octagonal rectangle-triangle random tiling}, J. Stat. Phys. {\bf 87}, 415 (1997).

\bibitem{gould93} M.D. Gould, {\it Quantum double finite group algebras and their representations}, Bull. Austral. Math. Soc. 
{\bf 48}, 275 (1993). 
 
\bibitem{dl09} K.A. Dancer and J. Links, {\it Universal spectral parameter-dependent Lax operators for the
Drinfeld double of the dihedral group $D_3$}, J. Phys A: Math. Theor. {\bf 42},  042002 (2009).

\bibitem{dfil09}  K.A. Dancer, P.E. Finch, P.S. Isaac, and J. Links, {\it Integrable boundary conditions for a non-Abelian anyon chain with $D(D_3)$ symmetry},
Nucl. Phys. B {\bf 812}, 456 (2009). 

\bibitem{lf} J. Links and A. Foerster, {\it On the construction of integrable closed chains with quantum supersymmetry}, J. Phys. A: Math. Gen. {\bf 30}, 2483 (1997).

\end{thebibliography}
\end{document}